\documentclass{article}
\usepackage{fullpage}
\usepackage{graphicx}
\usepackage{epstopdf}
\usepackage{latexsym}
\usepackage{amsmath}
\usepackage{amsfonts}
\usepackage{amssymb}
\usepackage{mathrsfs}
\usepackage{pifont}
\usepackage{yhmath}
\usepackage{undertilde}
\usepackage{bbm}
\usepackage{dsfont}
\usepackage{eufrak}
\usepackage{amsthm}
\usepackage[usenames,dvipsnames]{color}
\usepackage{stmaryrd}
\usepackage{wasysym}
\usepackage{bbm}
\usepackage[ruled]{algorithm2e}

\newtheorem{theorem}{Theorem}[section]

\newtheorem{conjecture}[theorem]{Conjecture}

\theoremstyle{remark}

\begin{document}
\newcounter{my}
\newenvironment{mylabel}
{
\begin{list}{(\roman{my})}{
\setlength{\parsep}{-1mm}
\setlength{\labelwidth}{8mm}
\usecounter{my}}
}{\end{list}}

\newcounter{my2}
\newenvironment{mylabel2}
{
\begin{list}{(\alph{my2})}{
\setlength{\parsep}{-0mm} \setlength{\labelwidth}{8mm}
\setlength{\leftmargin}{3mm}
\usecounter{my2}}
}{\end{list}}

\newcounter{my3}
\newenvironment{mylabel3}
{
\begin{list}{(\alph{my3})}{
\setlength{\parsep}{-1mm}
\setlength{\labelwidth}{8mm}
\setlength{\leftmargin}{10mm}
\usecounter{my3}}
}{\end{list}}

\vspace{-5em}

\title{\bf Approximation algorithms on $k-$ cycle covering and $k-$ clique covering}

\author{Zhongzheng Tang ${}^{b,c}$ \quad Zhuo Diao ${}^{a}$\thanks{Corresponding author. E-mail: diaozhuo@amss.ac.cn}
}
\date{
 ${}^a$ School of Statistics and Mathematics, Central University of Finance and Economics
Beijing 100081, China\\
${}^b$ Academy of Mathematics and Systems Science, Chinese Academy of Sciences\\
 Beijing 100190, China\\
 ${}^c$ School of Mathematical Sciences, University of Chinese Academy of Sciences\\
 Beijing 100049, China\\
}





\maketitle

\begin{abstract}
Given a weighted graph $G(V,E)$ with weight $\mathbf w: E\rightarrow Z^{|E|}_{+}$. A $k-$cycle covering is an edge subset $A$ of $E$ such that $G-A$ has no $k-$cycle. The minimum weight of $k-$cycle covering is the weighted covering number on $k-$cycle, denoted by $\tau_{k}(G_{w})$. In this paper, we design a $k-1/2$ approximation algorithm for the weighted covering number on $k-$cycle when $k$ is odd.

Given a weighted graph $G(V,E)$ with weight $\mathbf w: E\rightarrow Z^{|E|}_{+}$. A $k-$clique covering is an edge subset $A$ of $E$ such that $G-A$ has no $k-$clique. The minimum weight of $k-$clique covering is the weighted covering number on $k-$clique, denoted by $\widetilde{\tau_{k}}(G_{w})$. In this paper, we design a $(k^{2}-k-1)/2$ approximation algorithm for the weighted covering number on $k-$clique. Last, we discuss the relationship between $k-$clique covering and $k-$clique packing in complete graph $K_{n}$.


\end{abstract}


\section{k-cycle covering}

Given a weighted graph $G(V,E)$ with weight $\mathbf w: E\rightarrow Z^{|E|}_{+}$. A $k-$cycle covering is an edge subset $A$ of $E$ such that $G-A$ has no $k-$cycle. The problem of minimum weight of $k-$cycle covering can be described as follows:
\begin{eqnarray}
\tau_{k}(G_{w})=\min\{\mathbf w^T\mathbf x: A\mathbf x\geq\mathbf  1,\mathbf x\in\mathbf  \{0,~1\}\}\label{ILP1}
\end{eqnarray}

Where $A$ is a $k$-cycle-edge adjacent matrix and $\mathbf w$ is the weight vector. $\mathbf x$ represents the characteristic vector of edge.We obtain the relaxed programming of \eqref{ILP1} as follows:
\begin{eqnarray}
\min\{\mathbf w^T\mathbf x: A\mathbf x\geq\mathbf  1, 0 \leq\mathbf x\leq 1\}\label{LP1}
\end{eqnarray}
We compute the optimal solution $\hat{\mathbf x}^*$ of \eqref{LP1} in polynomial time, then we transfer $\hat{\mathbf x}^*$ to integral vector:
\begin{eqnarray}\label{eq1}
\mathbf x_e=
\begin{cases}
1   &\hspace{2cm} \hat{\mathbf x}^*_e\geq 1/k\\
0   &\hspace{2cm}o.w.
\end{cases}
\end{eqnarray}
Obviously, $\mathbf x$ is a feasible solution of ILP\eqref{ILP1} and $\mathbf w^T\mathbf x\leq k \mathbf w^T\hat{\mathbf x}^*$, which implies a $k$-approximation algorithm of minimum $k$-cycle covering problem.
\begin{algorithm} 
\KwIn{Weighted vector $\mathbf w$, $k$-cycle-edge adjacent matrix $A$}
\KwOut{A feasible solution $x$ of ILP\eqref{ILP1}, which reaches objective value no more than $k$ times of the optimal value.}
\begin{mylabel}
  \vspace{1.5mm}\item[1.] \hspace{3mm} Solve LP\eqref{LP1} and get the optimal solution $\hat{\mathbf x}^*$.
  \vspace{0mm}\item[2.]   \hspace{3mm} Compute $\mathbf x$ by equation\eqref{eq1}.
\vspace{-3mm}
\end{mylabel}
\caption{Approximation algorithm of minimum $k$-cycle covering} \label{alg1}
\end{algorithm}

\section{The $(k-\frac{1}{2})$-approximation algorithm when $k$ is odd}

We take advantage of specific strategy to reach better performance when $k$ is odd.

\begin{algorithm} 
\KwIn{Weighted graph $(G,\mathbf w)$}
\KwOut{A Edge set $E_k$, which covers every $k$-cycle in $G$.}
\begin{mylabel}
  \vspace{1.5mm}\item[0.] \hspace{3mm} Set $E_{k_1}=\emptyset$.
  \vspace{0mm}\item[1.] \hspace{3mm} Solve LP\eqref{LP1} and get the optimal solution $\hat{\mathbf x}^*$.
  \vspace{0mm}\item[2.]   \hspace{3mm} \textbf{For} every $e\in E$
  \vspace{0mm}\item[3.]   \hspace{7mm} \textbf{If} $\hat{\mathbf x}^*_e\geq 2/(2k-1)$~~~  \textbf{Then} $E_{k_1}=E_{k_1}\cup\{e\}$.
  \vspace{0mm}\item[4.]   \hspace{3mm} Suppose $E^\prime$ are these edges of all $k$-cycles in $G-E_{k_1}$ and let $G^\prime$ be a subgraph of $G$ induced by the edge set $E^\prime$.
  \vspace{0mm}\item[6.]   \hspace{3mm} Using the Greedy Algorithm or Random Algorithm, we can find an approximate solution of maximum weight bipartite graph
  \vspace{0mm}\item[7.]   \hspace{4mm}$B=(V_1,V_2,E_B)$, which satisfies $\mathbf W(E_B)\geq (1/2)\mathbf W(E^\prime)$. Set $E_{k_2}=E^\prime\setminus E_B$.
  \vspace{0mm}\item[8.]   \hspace{3mm} Output $E_k=E_{k_1}\cup E_{k_2}$.
  \vspace{-3mm}
\end{mylabel}
\caption{Approximation algorithm of minimum $k$-cycle covering} \label{alg2}
\end{algorithm}

\begin{theorem}
The Algorithm \ref{alg2} has $(k-\frac{1}{2})$ approximate ratio for the minimum $k$-cycle covering problem.
\end{theorem}
\begin{proof}
Suppose $\hat{\mathbf x}^*$ and $\mathbf x^*$ are the optimal solution of LP\eqref{LP1} and ILP\eqref{ILP1}, respectively.\\

Firstly, we indicate that $E_k$ is a $k$-cycle covering. Actually, for every $k$-cycle $C_{k}$ in $G$, if it doesn't contain any edge in $E_{k_1}$, then it is a $k$-cycle in $G-E_{k_1}$, thus it is a $k$-cycle in $G'$. Because $G^\prime-E_{k_2}$ is a bipartite graph, of course, $G^\prime-E_{k_2}$ has no $k$-cycle (here $k$ is odd). Thus $C_{k}$ contains some edge in $E_{k_2}$. Above all, we prove that $E_k$ is a $k$-cycle covering of $G$. \\

Additionally, we will show the approximate ratio.\\

On one hand, according to the rounding regulation, we know that:
\begin{eqnarray}\label{ineq1}
\sum\limits_{e\in E_{k_1}}\mathbf w_e\leq(k-\frac{1}{2})\sum\limits_{e\in E_{k_1}}\mathbf w_e\hat{\mathbf x}^*_e.
\end{eqnarray}
On the other hand, every $\hat{\mathbf x}^*_e$ related to $e\in E^\prime$ has the lower bound $1-2(k-1)/(2k-1)=1/(2k-1)$ on the grounds that there exists a $k$-cycle $C_k$ in $G^\prime$ containing $e$, satisfying $\hat{\mathbf x}^*_e\leq2/(2k-1)$ and $\sum_{e\in C_k}\hat{\mathbf x}^*_e\geq 1$.
\begin{eqnarray}\label{ineq2}
\sum\limits_{e\in E_{k_2}}\mathbf w_e\leq(1/2)\sum\limits_{e\in E^\prime}\mathbf w_e\leq(1/2)(2k-1)\sum\limits_{e\in E^\prime}\mathbf w_e \hat{\mathbf x}^*_e\leq (k-\frac{1}{2})\sum\limits_{e\in E\setminus E_{k_1}}\mathbf w_e\hat{\mathbf x}^*_e.
\end{eqnarray}
Combine inequalities \eqref{ineq1} and \eqref{ineq2}:
\begin{eqnarray}
\sum\limits_{e\in E_k}\mathbf w_e\leq(k-\frac{1}{2})\sum\limits_{e\in E}\mathbf w_e\hat{\mathbf x}^*_e\leq(k-\frac{1}{2})\sum\limits_{e\in E}\mathbf w_e \mathbf x^*_e.
\end{eqnarray}
which completes the proof.
\end{proof}

\section{The hardness of $k$-cycle covering when $k$ is even}

According to Algorithm \ref{alg1}, we trivially derive the $k$ approximate ratio whatever $k$ is odd or even. In previous section, we have shown $(k-\frac{1}{2})$ approximate ratio when $k$ is odd, but unfortunately we can't improve the approximate ratio when $k$ is even by using similar techniques. The following Theorem may tell us a possible reason and the hardness of the problem when $k$ is even.

\begin{theorem}(Paul Erd$\ddot{o}$s, Arthur Stone, 1946\cite{Erdos1946})\label{thm:ex}
The extremal function $ex(n; H)$ is defined to be the maximum number of edges in a graph of order $n$ not containing a subgraph isomorphic to $H$. \begin{eqnarray}
ex(n;H)=(\frac{r-2}{r-1}+o(1)){n\choose2}
\end{eqnarray}
where $r$ is the color number of $H$.
\end{theorem}

It is known that, when $H$ is bipartite, $ex(n; H) = o(n^2)$. Consider the special case, $H$ is an even cycle $C_{k}$, $ex(n; C_{k}) = o(n^2)$ thus $\tau_{k}(K_{n})= {n\choose2}-o(n^2)$. We have:

\begin{eqnarray}
lim_{n\rightarrow \infty} \tau_{k}(K_{n})/{n\choose2}=1
\end{eqnarray}

Thus there doesn't exist constant $0<c<1$ such that for every graph $G(V,E)$, $\tau_{k}(G)\leq cm$ holds on which is a key quality in our Algorithm \ref{alg2}.


\section{k-clique covering}

Given a weighted graph $G(V,E)$ with weight $\mathbf w: E\rightarrow Z^{|E|}_{+}$. A $k-$clique covering is an edge subset $A$ of $E$ such that $G-A$ has no $k-$clique. The problem of minimum weight of $k-$clique covering can be described as follows:
\begin{eqnarray}
\widetilde{\tau_{k}}(G_{w})=\min\{\mathbf w^T\mathbf x: A\mathbf x\geq\mathbf  1,\mathbf x\in\mathbf  \{0,~1\}\}\label{CILP1}
\end{eqnarray}

Where $A$ is a $k$-clique-edge adjacent matrix and $\mathbf w$ is the weight vector. $\mathbf x$ represents the characteristic vector of edge.We obtain the relaxed programming of \eqref{ILP1} as follows:
\begin{eqnarray}
\min\{\mathbf w^T\mathbf x: A\mathbf x\geq\mathbf  1, 0 \leq\mathbf x\leq 1\}\label{CLP1}
\end{eqnarray}
We compute the optimal solution $\hat{\mathbf x}^*$ of \eqref{CLP1} in polynomial time, then we transfer $\hat{\mathbf x}^*$ to integral vector:
\begin{eqnarray}\label{ceq1}
\mathbf x_e=
\begin{cases}
1   &\hspace{2cm} \hat{\mathbf x}^*_e\geq 1/{k\choose2}\\
0   &\hspace{2cm}o.w.
\end{cases}
\end{eqnarray}
Obviously, $\mathbf x$ is a feasible solution of ILP\eqref{CILP1} and $\mathbf w^T\mathbf x\leq {k\choose2} \mathbf w^T\hat{\mathbf x}^*$, which implies a ${k\choose2}$-approximation algorithm of minimum $k$-clique covering problem.
\begin{algorithm} 
\KwIn{Weighted vector $\mathbf w$, $k$-clique-edge adjacent matrix $A$}
\KwOut{A feasible solution $x$ of ILP\eqref{CILP1}, which reaches objective value no more than ${k\choose2}$ times of the optimal value.}
\begin{mylabel}
  \vspace{1.5mm}\item[1.] \hspace{3mm} Solve LP\eqref{CLP1} and get the optimal solution $\hat{\mathbf x}^*$.
  \vspace{0mm}\item[2.]   \hspace{3mm} Compute $\mathbf x$ by equation\eqref{ceq1}.
\vspace{-3mm}
\end{mylabel}
\caption{Approximation algorithm of minimum $k$-clique covering} \label{alg3}
\end{algorithm}

\section{The $(k^2-k-1)/2$-approximation algorithm for minimum $k$-clique covering problem}

Similarly with Algorithm \ref{alg2}, we have the following approximation algorithm for minimum $k$-clique covering problem.

\begin{algorithm} 
\KwIn{Weighted graph $(G,\mathbf w)$}
\KwOut{A Edge set $E_k$, which covers every $k$-clique in $G$.}
\begin{mylabel}
  \vspace{1.5mm}\item[0.] \hspace{3mm} Set $E_{k_1}=\emptyset$.
  \vspace{0mm}\item[1.] \hspace{3mm} Solve LP\eqref{LP1} and get the optimal solution $\hat{\mathbf x}^*$.
  \vspace{0mm}\item[2.]   \hspace{3mm} \textbf{For} every $e\in E$
  \vspace{0mm}\item[3.]   \hspace{7mm} \textbf{If} $\hat{\mathbf x}^*_e\geq 2/(2{k\choose2}-1)$~~~  \textbf{Then} $E_{k_1}=E_{k_1}\cup\{e\}$.
  \vspace{0mm}\item[4.]   \hspace{3mm} Suppose $E^\prime$ are these edges of all $k$-cliques in $G-E_{k_1}$ and let $G^\prime$ be a subgraph of $G$ induced by the edge set $E^\prime$.


  \vspace{0mm}\item[6.]   \hspace{3mm} Using the Greedy Algorithm or Random Algorithm, we can find an approximate solution of maximum weight bipartite graph
  \vspace{0mm}\item[7.]   \hspace{4mm}$B=(V_1,V_2,E_B)$, which satisfies $\mathbf W(E_B)\geq (1/2)\mathbf W(E^\prime)$. Set $E_{k_2}=E^\prime\setminus E_B$.
  \vspace{0mm}\item[8.]   \hspace{3mm} Output $E_k=E_{k_1}\cup E_{k_2}$.
  \vspace{-3mm}
\end{mylabel}
\caption{Approximation algorithm of $k$-clique covering} \label{alg4}
\end{algorithm}

\begin{theorem}
The Algorithm \ref{alg4} has $(k^2-k-1)/2$ approximate ratio for the minimum $k$-clique covering problem.
\end{theorem}

\begin{proof}

Suppose $\hat{\mathbf x}^*$ and $\mathbf x^*$ are the optimal solution of LP\eqref{CLP1} and ILP\eqref{CILP1}, respectively.\\

Firstly, we indicate that $E_k$ is a $k$-clique covering. Actually, for every $k$-clique $K_{k}$ in $G$, if it doesn't contain any edge in $E_{k_1}$, then it is a $k$-clique in $G-E_{k_1}$, thus it is a $k$-clique in $G'$. Because $G^\prime-E_{k_2}$ is a bipartite graph, of course, $G^\prime-E_{k_2}$ has no $k$-clique ($G^\prime-E_{k_2}$ has no triangle). Thus $K_{k}$ contains some edge in $E_{k_2}$. Above all, we prove that $E_k$ is a $k$-clique covering of $G$. \\

Additionally, we will show the approximate ratio.\\

On one hand, according to the rounding regulation, we know that:
\begin{eqnarray}\label{ineq1}
\sum\limits_{e\in E_{k_1}}\mathbf w_e\leq({k\choose2}-\frac{1}{2})\sum\limits_{e\in E_{k_1}}\mathbf w_e\hat{\mathbf x}^*_e.
\end{eqnarray}
On the other hand, every $\hat{\mathbf x}^*_e$ related to $e\in E^\prime$ has the lower bound $1-2({k\choose2}-1)/(2{k\choose2}-1)=1/(2{k\choose2}-1)$ on the grounds that there exists a $k$-clique $K_k$ in $G^\prime$ containing $e$, satisfying $\hat{\mathbf x}^*_e\leq2/(2{k\choose2}-1)$ and $\sum_{e\in K_k}\hat{\mathbf x}^*_e\geq 1$.
\begin{eqnarray}\label{ineq2}
\sum\limits_{e\in E_{k_2}}\mathbf w_e\leq(1/2)\sum\limits_{e\in E^\prime}\mathbf w_e\leq(1/2)(2{k\choose2}-1)\sum\limits_{e\in E^\prime}\mathbf w_e \hat{\mathbf x}^*_e\leq ({k\choose2}-\frac{1}{2})\sum\limits_{e\in E\setminus E_{k_1}}\mathbf w_e\hat{\mathbf x}^*_e.
\end{eqnarray}
Combine inequalities \eqref{ineq1} and \eqref{ineq2}:
\begin{eqnarray}
\sum\limits_{e\in E_k}\mathbf w_e\leq({k\choose2}-\frac{1}{2})\sum\limits_{e\in E}\mathbf w_e\hat{\mathbf x}^*_e\leq({k\choose2}-\frac{1}{2})\sum\limits_{e\in E}\mathbf w_e \mathbf x^*_e.
\end{eqnarray}
which completes the proof.

\end{proof}

\section{$k$-clique covering and $k$-clique packing in $K_n$}

Given a graph $G(V,E)$, a $k-$clique packing is a set of edge-disjoint $k-$cliques in $G$. The problem of maximum number of $k-$clique packing can be described as follows:

\begin{eqnarray}
\widetilde{\nu_{k}}(G)=\max\{\mathbf 1^T\mathbf x: A^T\mathbf y\leq\mathbf  1,\mathbf y\in\mathbf  \{0,~1\}\}
\end{eqnarray}

It is easy to see for every graph $G$, $\widetilde{\nu_{k}}(G)\leq\widetilde{\tau_{k}}(G)\leq{k\choose2}\widetilde{\nu_{k}}(G)$ holds on.\\

According to Theorem \ref{thm:ex}, the $k$-clique covering number of $K_n$ is $\widetilde{\tau_{k}}(K_{n})=(1/(k-1)-o(1)){n\choose2}$.\\

As for the packing number, we need the classical results in Block Design Theory.\\

A $2$-design (or BIBD, standing for balanced incomplete block design),denoted by $(v,k,\lambda)$-BIBD, is a family of $k-$ element subsets of $X$, called blocks, such that any pair of distinct points $x$ and $y$ in $X$ is contained in $\lambda$ blocks. Here $v$ is number of points, number of elements of $X$, $k$ is number of points in a block, $\lambda$ is number of blocks containing any two distinct points. We have next famous theorem:

\begin{theorem}(\cite{Richard1975})\label{thm:Block}
Given positive integers $k$ and $\lambda$, $(v,k,\lambda)$-BIBD exist for all sufficiently large integers $v$ for which the congruences
$\lambda(v-1)\equiv0(\hspace{-2mm}\mod k-1)$ and $\lambda v(v-1)\equiv0(\hspace{-2mm}\mod k(k-1))$ are valid.
\end{theorem}

When $\lambda=1$, it is easy to see  $(n,k,1)$-BIBD exists if and only if $K_n$ contains a perfect $k$-clique packing, which is a $k$-clique packing such that every edge belongs to a $k$-clique. Thus the above Theorem \ref{thm:Block} is equivalent to the following Theorem:

\begin{theorem}\label{thm:packing}
For all sufficiently large integers $n$ satisfying $n\equiv1,k~(\hspace{-2mm}\mod k(k-1))$, then $K_n$ contains a perfect $k$-clique packing.
\end{theorem}

For all sufficiently large integers $n$ satisfying $n\equiv1,k~(\hspace{-2mm}\mod k(k-1))$, $K_n$ contains $\frac{n(n-1)}{k(k-1)}$ edge-disjoint $k$-clique. So we have $\widetilde{\nu_{k}}(K_{n})\sim\sim\frac{n(n-1)}{k(k-1)}$ and
$k$-clique covering number over $k$-clique packing number in $K_n$ is $k/2$ when $n\rightarrow\infty$, that is:

\begin{eqnarray}
lim_{n\rightarrow \infty} \widetilde{\tau_{k}}(K_{n})/\widetilde{\nu_{k}}(K_{n})= lim_{n\rightarrow \infty} (1/(k-1)-o(1)){n\choose2}/\frac{n(n-1)}{k(k-1)}=k/2
\end{eqnarray}

Recall Tuza's Conjecture, which is related to the ratio of triangle covering number and triangle packing number:
\begin{conjecture}(Tuza, 1981\cite{tuza1981})
$\tau(G)\leq2\nu(G)$ holds for every graph $G$.
\end{conjecture}

For the ratio of $k$-clique covering number and $k$-clique packing number in graph $G$, the trivial upper bound is ${k\choose2}$. We guess there exists a upper bound between $k/2$ and ${k\choose2}$ for every graph $G$.

\bibliography{ref}

\end{document}